\documentclass[11pt]{amsart}

\usepackage[T2A,T1]{fontenc}
\usepackage[utf8]{inputenc}
\usepackage[main=english,russian]{babel}
\usepackage{microtype}
\usepackage{amsmath,amssymb,amsthm}
\usepackage{graphicx}
\graphicspath{{figures/}}
\usepackage[margin=1in]{geometry}
\usepackage[numbers,sort&compress]{natbib}
\usepackage[colorlinks=true,linkcolor=blue,citecolor=blue,urlcolor=blue]{hyperref}

\numberwithin{equation}{section}
\newtheorem{theorem}{Theorem}[section]
\newtheorem{proposition}[theorem]{Proposition}
\newtheorem{lemma}[theorem]{Lemma}
\newtheorem{corollary}[theorem]{Corollary}
\theoremstyle{remark}
\newtheorem{remark}[theorem]{Remark}

\renewcommand{\C}{\mathbb C}
\newcommand{\E}{\mathbb E}
\renewcommand{\Prob}{\mathbb P}
\newcommand{\Poly}{\operatorname{Poly}}
\newcommand{\DPP}{\operatorname{DPP}}
\newcommand{\capc}{\operatorname{cap}}
\newcommand{\dist}{\operatorname{dist}}
\newcommand{\Tr}{\operatorname{Tr}}

\title[Quenched Growth of Normal Matrices]{Quenched Growth of Normal Matrices and Harmonic Measure}

\hypersetup{pdftitle={Quenched Growth of Normal Matrices and Harmonic Measure}}

\author{Oleg Alekseev}\thanks{ORCID: \href{https://orcid.org/0000-0002-0666-601X}{0000-0002-0666-601X}.}

\address{Chebyshev Laboratory, Department of Mathematics and Computer Science, Saint Petersburg State University, Saint Petersburg, Russia}

\address{Department of Mathematics, HSE University (National Research University Higher School of Economics), Moscow, Russia}


\date{}

\begin{document}
\begin{abstract}
	The eigenvalues of a large normal random matrix form a two-dimensional
	Coulomb gas.  At fixed weight, the spectral droplet grows with the matrix
	rank, and its semiclassical density increment is harmonic measure.
	We freeze an $N$-eigenvalue configuration and sample $M$ further
	eigenvalues using the projection onto polynomials of degree less than $N+M$
	that vanish at the frozen points.
	Their conditional law is an exact polynomial ensemble in the potential
	of the frozen charges.
	For a regular analytic droplet, we prove that the spatial distribution of the new eigenvalues converges to harmonic
	measure on the old boundary when $M/\log N\to\infty$ and $M/N\to0$.
	This convergence holds conditionally on the frozen spectrum, outside a set of frozen configurations whose probability tends to zero.
	The equilibrium process formed from
	orthogonal polynomials of degrees $N,\ldots,N+M-1$ has the same limit.
	An exact finite-rank identity relates polynomial moments of the
	conditional mean to fluctuations of the frozen spectrum, with a remainder
	from polynomial modes of degree $N+M$ and above.
\end{abstract}

\maketitle

\section{Introduction}

The eigenvalues of a normal random matrix form a two-dimensional Coulomb
gas after the matrix angular variables are integrated out~\cite{ChauZaboronsky1998,TeodorescuNPB,ElbauFelder}.  At large matrix rank,
they fill a compact region of the plane, called the spectral droplet.
In the large matrix limit, increasing the rank at fixed weight drives
Laplacian growth of the droplet boundary \cite{ABWZ02,HedenmalmMakarov}.
This evolution can be described through harmonic moments and conformal maps
\cite{MineevWeinsteinWiegmannZabrodin,WiegmannZabrodinConformal}.

We compare two ways of adding eigenvalues.  In \emph{equilibrium growth},
we compare the rank-$N$ and rank-$(N+M)$ ensembles at the same weight, such that
all eigenvalues are sampled anew.  The increase in mean density comes
from the orthogonal polynomials of degrees $N,\ldots,N+M-1$, which we
call the \emph{orbital band}.  For a regular analytic droplet, the band's
normalized mean density tends to harmonic measure when $M/N\to0$.
If also $M\to\infty$, a batch sampled from this band has the same limit
in probability.

In \emph{quenched growth}, we hold the observed spectrum
$Y=\{y_1,\ldots,y_N\}$ fixed and sample only the $M$ new eigenvalues.
Their conditional law is defined by the projection onto polynomials
of degree less than $N+M$ that vanish at the frozen points.
The old spectrum thus acts as frozen disorder.  The joint density of
the new points contains the factor
\[
	\prod_{a=1}^M\prod_{i=1}^N|z_a-y_i|^2.
\]
This factor describes their interaction with the old charges.
The new points also repel one another through the Vandermonde determinant
in their density.

Our main result concerns a single quenched batch.  As $N\to\infty$,
let its size $M=M_N$ satisfy
\[
	\log N\ll M\ll N.
\]
For a regular analytic droplet, the spatial distribution of the new points
converges weakly to exterior harmonic measure on the old boundary.
The limit holds conditionally on the old spectrum, with the probability
of exceptional frozen configurations tending to zero.
Thus, the two growth rules have the same limiting boundary distribution,
although their laws differ at finite $N$.

These differences are visible in the numerical three-fold example of
Section~\ref{sec:examples}.  After projection onto the boundary,
quenched batches show weaker angular repulsion and larger fluctuations
in arc counts and low angular Fourier modes than the equilibrium band.
We separate the structure factor into the contribution of the conditional
mean and the fluctuations of batches sampled at fixed old spectrum.

Earlier work established the harmonic-measure description of equilibrium
growth.  Agam, Bettelheim, Wiegmann, and Zabrodin obtained this result
for a newly occupied orbital in ``one-electron'' Laplacian growth
\cite{ABWZ02}, and Hedenmalm and Makarov proved a weak derivative formula
for the equilibrium measure of a regular growing droplet
\cite{HedenmalmMakarov}.  Our theorem gives a conditional counterpart
to this equilibrium law when the observed spectrum is kept fixed.

The construction also connects growth to conditional point processes.
For Ginibre-type processes, conditional laws have explicit
orthogonal-polynomial weights
\cite{BufetovQiuConditional,BufetovQiuShamov}.
In the Ginibre process, the exterior configuration determines how many
points lie inside a bounded region, while their positions remain random
(number rigidity and tolerance) \cite{GhoshPeres}.
Here polynomial projections give an exact finite-$N$ law for adding
a prescribed number of points to an observed spectrum.

For equilibrium ensembles, macroscopic connected correlations are
related to the exterior Dirichlet problem
\cite{WiegmannZabrodinCorrelations}, and smooth bulk linear statistics
have Gaussian limits under appropriate hypotheses
\cite{AmeurHedenmalmMakarovFluctuations}.
Boundary contributions are known for the Ginibre disk \cite{RiderVirag},
and a Gaussian free field with free boundary conditions describes
fluctuations for regular droplets \cite{AmeurHedenmalmMakarovWard}.
Proposition~\ref{prop:screening} relates conditional mean polynomial
moments of the new batch to fluctuations of the old spectrum and a
finite-rank remainder.

In diffusion-limited aggregation, harmonic measure gives the attachment
distribution of a Brownian particle arriving at a frozen aggregate
\cite{WittenSanderPRL}.
Stochastic Laplacian growth has also been studied through independent
Brownian batches \cite{AlekseevMineevWeinstein} and iterated conformal
maps \cite{HastingsLevitov,NorrisTurner}.

The paper is organized as follows. Section~\ref{sec:rules} defines the two growth rules.
Section~\ref{sec:theorem} states the quenched harmonic-measure theorem,
and Section~\ref{sec:proof} gives its proof.
Section~\ref{sec:examples} compares the two rules numerically on the
three-fold droplet.  Section~\ref{sec:discussion} discusses the results
and open questions.  Appendix~\ref{app:equilibrium-cue} proves the
fixed-batch CUE limit for the orbital band.

\section{The normal-matrix ensemble and two growth rules}\label{sec:rules}

We work in the complex plane with normalized area measure
$dA(z)=\pi^{-1}dx\,dy$ with $z=x+iy$.  Let $Q:\C\to\mathbb R$ be a confining
potential and set
\[
	d\lambda_N(z)=e^{-2NQ(z)}dA(z).
\]
We write $\Poly_{<n}$ for the space of complex polynomials of degree less than
$n$, for $n\ge1$ such that this space lies in $L^2(\lambda_N)$.
The weight $e^{-2NQ}dA$ stays fixed as $n$ changes. Below, we write
\[
	V(z_1,\ldots,z_n)=\prod_{1\le a<b\le n}(z_b-z_a)
\]
for the Vandermonde determinant. Then, the rank-$n$ normal-matrix eigenvalue law
at the fixed weight reads
\[
	\Prob_n(dz_1\cdots dz_n)
	=\frac{|V(z_1,\ldots,z_n)|^2}{\mathcal Z_{n,N}}
	\prod_{a=1}^n d\lambda_N(z_a).
\]
Here $\mathcal Z_{n,N}$ is the normalization constant, also called the
partition function.

This density is also the Gibbs law of a planar Coulomb gas at inverse temperature
$2$, with logarithmic pair repulsion and confining potential $NQ$.

By $p_{j,N}$ we denote the monic polynomial of degree $j$ orthogonal in
$L^2(\lambda_N)$, and define
\begin{equation}\label{eq:orthogonal-polynomial}
	h_{j,N}=\|p_{j,N}\|_{L^2(\lambda_N)}^2,
	\qquad \phi_{j,N}=p_{j,N}/\sqrt{h_{j,N}}.
\end{equation}
Orthogonality gives the partition function as a product of these squared norms:
\begin{equation}\label{eq:tau-normalizer}
	\mathcal Z_{n,N}
	=\int_{\C^n}|V(z_1,\ldots,z_n)|^2
	\prod_{a=1}^n d\lambda_N(z_a)
	=n!\prod_{j=0}^{n-1}h_{j,N}.
\end{equation}
This partition function also appears in the tau-function formulation of
normal random matrices \cite{KostovKricheverMineevWeinsteinWiegmannZabrodin}.

We denote by $\Pi_n$ the orthogonal projection from $L^2(\lambda_N)$ onto
$\Poly_{<n}$, and recall its integral kernel relative to $d\lambda_N$:
\[
	K_n(z,w)=\sum_{j=0}^{n-1}\phi_{j,N}(z)\overline{\phi_{j,N}(w)},
	\qquad
	(\Pi_n f)(z)=\int_{\C}K_n(z,w)f(w)\,d\lambda_N(w).
\]

The orthogonal polynomial identity for the Vandermonde determinant gives the equivalent
determinantal form of the eigenvalue law:
\begin{equation}\label{eq:dpp}
	\Prob_n(dz_1\cdots dz_n)
	=\frac{1}{n!}\det[K_n(z_a,z_b)]_{a,b=1}^n
	\prod_{a=1}^n d\lambda_N(z_a).
\end{equation}
We will refer to this projection process as $\DPP(\Pi_n)$ (see \cite{HoughKrishnapurPeresVirag}
for background on determinantal point
processes).
Its one-point intensity relative to $dA$ reads
\[
	R_{n,N}(z)=K_n(z,z)e^{-2NQ(z)}.
\]

Below, we put $H=\Poly_{<N}$ and $F=\Poly_{<N+M}$, and write
$\Pi_H=\Pi_N$, $\Pi_F=\Pi_{N+M}$ for the corresponding projections and
$K_H=K_N$, $K_F=K_{N+M}$ for their kernels.

\subsection{Equilibrium growth: newly occupied orbitals}
Equilibrium growth compares the two canonical ensembles $\DPP(\Pi_H)$ and
$\DPP(\Pi_F)$.  Since $H\subset F$, the operator
\begin{equation}\label{eq:band-projection}
	\Pi^{\mathrm{band}}_{N,M}
	=\Pi_F-\Pi_H
	=\sum_{j=N}^{N+M-1}\phi_{j,N}\otimes\phi_{j,N}
\end{equation}
is the rank-$M$ orthogonal projection onto the newly occupied orbital band,
and $(\phi\otimes\phi)g=\langle g,\phi\rangle_{L^2(\lambda_N)}\phi$.
We describe equilibrium growth by the point process on this band and define
its empirical measure  by
\begin{equation}\label{eq:equilibrium-band-process}
	Z^{\mathrm{eq}}_{N,M}\sim\DPP(\Pi^{\mathrm{band}}_{N,M}),
	\qquad
	\widehat\mu^{\mathrm{eq}}_{N,M}
	=\frac1M\sum_{z\in Z^{\mathrm{eq}}_{N,M}}\delta_z .
\end{equation}
Its mean is the increase in equilibrium density divided by the number
of added points, namely,
\begin{equation}\label{eq:equilibrium-band-mean}
	\mu^{\mathrm{eq}}_{N,M}
	=\E\bigl[\widehat\mu^{\mathrm{eq}}_{N,M}\bigr]
	=\frac{R_{N+M,N}-R_{N,N}}{M}\,dA.
\end{equation}
Hence, the one-point intensity of equilibrium growth reads
\begin{equation}\label{eq:band-density}
	R_{N+M,N}(z)-R_{N,N}(z)
	=e^{-2NQ(z)}\sum_{j=N}^{N+M-1}|\phi_{j,N}(z)|^2.
\end{equation}

The boundary-layer calculation of \cite[Secs.~3--4]{AlekseevTMF}
gives a circular unitary ensemble (CUE) law for the joint distribution
of the conformal angles of the $M$ points sampled from the equilibrium
orbital band.  This law describes their angular correlations, and each
individual angle is uniformly distributed.
The following proposition proves convergence to this joint angular law
as $N\to\infty$ with $M$ fixed, under regular analytic assumptions.
\begin{proposition}[CUE limit for equilibrium growth]
	\label{prop:radial-band-cue}
	Let $Q$ satisfy the regular analytic assumptions of
	Section~\ref{sec:theorem}. We fix $M$ and introduce  $\varphi$ as the normalized
	exterior conformal map of $D$.  In a fixed collar of $\partial D$ where
	$\varphi$ extends conformally, we set $\Theta(z)=\arg\varphi(z)$ and
	extend $\Theta$ measurably elsewhere.
	Let us label the points of $Z^{\mathrm{eq}}_{N,M}$ exchangeably and set
	$\theta_a=\Theta(z_a)$. Then, as $N\to\infty$, the joint law of the labelled
	angles converges in total variation
	to the CUE eigenangle density of Haar-distributed unitary matrices
	\cite{DysonJMP1},
	\begin{equation}\label{eq:radial-band-cue}
		\frac{1}{M!(2\pi)^M}
		\prod_{1\le a<b\le M}|e^{i\theta_b}-e^{i\theta_a}|^2
		\prod_{a=1}^M d\theta_a .
	\end{equation}
\end{proposition}
In Appendix~\ref{app:equilibrium-cue} we prove this limit by comparing adjacent
orbitals near the boundary.

Let us denote by $D_t$ the spectral droplet at
filling $t$ in the regular semiclassical regime, such that
\begin{equation}\label{eq:equilibrium-measure}
	\mu_t=2\Delta Q\,\mathbf1_{D_t}dA,
	\qquad \Delta=\partial_z\partial_{\bar z},
	\qquad \mu_t(\C)=t,
\end{equation}
where $\mu_t$ is the large-$N$ equilibrium measure with support $D_t$.
Under the regularity assumptions of Section~\ref{sec:theorem}, its growth
obeys
\begin{equation}\label{eq:rule-a-drift}
	\frac{d\mu_t}{dt}=\omega_{D_t},
\end{equation}
where $\omega_{D_t}$ is exterior harmonic measure on $\partial D_t$, viewed
from infinity.  Hedenmalm and Makarov prove the right-hand weak
derivative for regular droplet growth
\cite[Proposition~5.20]{HedenmalmMakarov}.  Hedenmalm and Wennman
state the weighted growth law for a smooth regular flow~\cite[Section~2.3]{HedenmalmWennman}.

\subsection{Quenched growth: Coulomb gas in a frozen eigenvalue field}
In this case the old spectrum is sampled first,
$Y=\{y_1,\ldots,y_N\}\sim\DPP(\Pi_H)$. We denote its law and expectation
by $\Prob_N^{\mathrm{old}}$ and $\E_N^{\mathrm{old}}$, and
abbreviate them as $\Prob_Y$ and $\E_Y$ in the proof.
For distinct $Y$, the functions $K_F(\cdot,y_i)$, $1\le i\le N$, represent
evaluation at the frozen points.  Let $P_Y$ be the orthogonal projection
in $F$ onto their span, and define
\begin{equation}\label{eq:completion}
	\Pi_Y^\perp=\Pi_F-P_Y,
	\qquad Z\mid Y\sim\DPP(\Pi_Y^\perp).
\end{equation}
For fixed $Y$, we denote this conditional law and expectation by
$\Prob_{N,M}^{Y}$ and $\E_{N,M}^{Y}$.  When $N$ and $M$ are clear, we
also write $\Prob(\,\cdot\mid Y)$ and $\E[\,\cdot\mid Y]$.
We keep $Y$ fixed and sample
$Z=\{z_1,\ldots,z_M\}$ from this projection process.

The following proposition gives an explicit joint density for the newly sampled
points.  The frozen eigenvalues enter its weight through a polynomial
whose zeros are precisely the points of $Y$.

\begin{proposition}[Exact polynomial form]\label{prop:polynomial}
	For every distinct $Y$,
	\begin{equation}\label{eq:q-space}
		\operatorname{Ran}\Pi_Y^\perp=q_Y\Poly_{<M},
		\qquad q_Y(z)=\prod_{i=1}^N(z-y_i).
	\end{equation}
	Let $G_Y$ be the $M\times M$ Gram matrix of the ordered basis
	$q_Y,q_Yz,\ldots,q_Yz^{M-1}$ in $L^2(\lambda_N)$.  Conditional on $Y$,
	\begin{equation}\label{eq:conditional-density}
		\Prob(dZ\mid Y)=
		\frac{|V(Z)|^2\prod_{a=1}^M|q_Y(z_a)|^2d\lambda_N(z_a)}
		{M!\det G_Y}.
	\end{equation}
	For $M=1$, this becomes
	\begin{equation}\label{eq:one-particle}
		\Prob(dz\mid Y)=
		\frac{|q_Y(z)|^2e^{-2NQ(z)}dA(z)}
		{\int_{\C}|q_Y(w)|^2e^{-2NQ(w)}dA(w)}.
	\end{equation}
\end{proposition}
\begin{proof}
	A polynomial in $F$, where $F=\Poly_{<N+M}$, is orthogonal to every $K_F(\cdot,y_i)$ precisely when
	it vanishes at every $y_i$.  Polynomial division gives
	\eqref{eq:q-space}, and for an arbitrary basis
	$f_0,\ldots,f_{M-1}$, the projection-DPP density reads
	$|\det[f_{j-1}(z_a)]|^2/(M!\det[\langle f_i,f_j\rangle])$.
	Upon taking $f_j=q_Yz^j$, we obtain exactly \eqref{eq:conditional-density}, and its rank-one
	case reduces to \eqref{eq:one-particle}.
\end{proof}

The conditional density \eqref{eq:conditional-density} describes a Coulomb
gas in the field of the frozen eigenvalues.  In exponential form,
\[
	\Prob(dZ\mid Y)\propto e^{-\mathcal H_Y(Z)}\prod_{a=1}^M dA(z_a),
\]
where the conditional Coulomb-gas energy is
\[
	\mathcal H_Y(Z)=2N\sum_{a=1}^M Q(z_a)
	-2\sum_{a=1}^M\sum_{i=1}^N\log|z_a-y_i|
	-2\sum_{1\le a<b\le M}\log|z_a-z_b|.
\]
The new particles are confined by $Q$ and repel both the frozen charges
and one another.

When $N$ eigenvalues are fixed at the distinct points $Y$,
\eqref{eq:completion} gives the conditional law of the remaining
$M$ eigenvalues in the rank-$(N+M)$ ensemble.
The conditional kernel is the orthogonal projection onto polynomials
in $F$ that vanish at the frozen points $Y$ \cite{ShiraiTakahashiPalm}.
For labelled old and new points, the interaction factors in the two
joint densities cancel, so that their ratio is
\[
	\frac{d(\Prob_N^{\mathrm{old}}(dY)\Prob(dZ\mid Y))}
	{d\Prob_{N+M}(dY,dZ)}
	=\binom{N+M}{N}
	\frac{\prod_{j=N}^{N+M-1}h_{j,N}}{\det G_Y}.
\]
Thus, the two laws share the same
conditional distribution of the new points and differ only in the
probabilities assigned to the old spectrum.
By \eqref{eq:gram-lower} below, the ratio never exceeds
$\binom{N+M}{N}$, and any event therefore has quenched probability at most
this factor times its probability in the rank-$(N+M)$ equilibrium ensemble.

Adding $M$ eigenvalues shifts the equilibrium boundary by order $M/N$.
To compare the two growth processes, we must also consider the normal
width $N^{-1/2}$ of the polynomial modes and the typical spacing $1/M$
between boundary points in conformal angle.
The boundary shift tends to zero when $M/N\to0$, but its size relative
to the other scales depends on the batch size.  For $M\ll\sqrt N$,
we have $M/N\ll N^{-1/2}\ll 1/M$, while for $M\gg\sqrt N$, the order is
reversed.  The three scales are comparable at $M\sim\sqrt N$.
The squared ratio of the boundary shift to the normal width is $M^2/N$,
which also appears in the radial equilibrium-band calculation
\eqref{eq:radial-band-cue-correction}.

The equilibrium band provides a reference for quenched growth.
The following proposition shows that its mean spatial distribution
converges to harmonic measure on the old boundary when $M/N\to0$.

\begin{proposition}[Harmonic-measure limit of the orbital band]
	\label{prop:band-harmonic-measure}
	Suppose $Q$ satisfies the regular analytic assumptions of
	Section~\ref{sec:theorem}, and let $M=M_N\geq1$ satisfy $M/N\to0$.
	Then, $\mu^{\mathrm{eq}}_{N,M}$ converges weakly to $\omega_D$.
	If also $M\to\infty$, the empirical measure
	$\widehat\mu^{\mathrm{eq}}_{N,M}$ converges weakly to $\omega_D$
	in probability.
\end{proposition}
\begin{proof}
	Let us write $\nu_{j,N}=|\phi_{j,N}|^2d\lambda_N$ and set $t_j=j/N$.
	The uniform $L^2$ orbital expansion of
	\cite[Theorem~3.2]{HedenmalmWennman}, together with its leading
	coefficient \cite[Theorem~1.5]{HedenmalmWennman}, applies for
	$t_j$ in a fixed neighborhood of $1$.  The $L^2$ error gives an $o(1)$
	error for bounded test functions,
	uniformly in $t_j$.  In the exterior conformal coordinate for
	$D_{t_j}$, the leading weighted density is a Gaussian ridge of
	width $O(N^{-1/2})$ about the unit circle.  Uniform radial Laplace
	integration with the explicit leading coefficient gives
	$d\theta/(2\pi)$, while the mass outside a fixed boundary collar
	is negligible.  Thus,
	for every bounded continuous $f$,
	\[
		\sup_{N\leq j<N+M}\left|
		\int f\,d\nu_{j,N}-\int f\,d\omega_{D_{t_j}}
		\right|\longrightarrow0.
	\]
	Here $\sup_{N\leq j<N+M}|t_j-1|\leq M/N\to0$.  The smooth
	droplet flow makes $\omega_{D_{t_j}}\Rightarrow\omega_D$
	uniformly over this range.  Averaging $\nu_{j,N}$ and using
	\eqref{eq:equilibrium-band-mean} proves the mean-measure assertion.

	For real bounded $f$, the rank-$M$ projection-DPP variance identity
	gives
	\[
		\operatorname{Var}\left(\int f\,d\widehat\mu^{\mathrm{eq}}_{N,M}\right)
		\leq\frac{\|f\|_\infty^2}{M}.
	\]
	As $M\to\infty$, Chebyshev's inequality shows that each sampled average
	concentrates around its mean.  Since the mean converges to the
	harmonic-measure average, the empirical measure converges to $\omega_D$
	in probability.  Here weak convergence in probability means that
	$\int f\,d\widehat\mu^{\mathrm{eq}}_{N,M}\to\int f\,d\omega_D$ in
	probability for each bounded continuous $f$.
\end{proof}

\section{Quenched harmonic measure for one batch}\label{sec:theorem}

We now turn to a single batch of eigenvalues added to a frozen spectrum.
Below we show that for a regular analytic droplet and $\log N\ll M\ll N$, their spatial
distribution converges to exterior harmonic measure on the old boundary,
just as for equilibrium growth.  This convergence holds conditionally
on the old spectrum, with the probability of exceptional frozen
configurations tending to zero.

We assume that $Q$ is $1$-admissible in the sense of
\cite[Definition~1.2]{HedenmalmWennman}.  Thus, $Q$ is globally $C^2$ and
is real analytic and strictly subharmonic near its unit-filling droplet.
At infinity, we require $\liminf_{|z|\to\infty}Q(z)/\log|z|>1$.
The droplet has a smooth simple closed boundary and equals the full
coincidence set of the obstacle solution.
These hypotheses give, for $t$ in a fixed neighborhood of $1$, a
regular flow with simple analytic boundaries and
$D_t=\{z:\widehat Q_t(z)=Q(z)\}$, where $\widehat Q_t$ is the
obstacle solution \cite[Sections~1.5 and 2.1]{HedenmalmWennman}.
Thus, $D=D_1$ is a compact regular Jordan droplet.  Its exterior conformal map
is normalized by
\begin{equation}\label{eq:conformal-map}
	\varphi:\C\setminus D\longrightarrow\C\setminus\overline{\mathbb D},
	\qquad \varphi(z)=z/r+O(1),\quad
	r=\capc(D)>0.
\end{equation}
The constant $r$ is the logarithmic capacity, or exterior conformal radius,
of $D$.  Harmonic measure $\omega_D$ is obtained by mapping the uniform
angular measure $d\theta/(2\pi)$ to $\partial D$ using $\varphi^{-1}$.

The following theorem makes this limit precise: averages over the new
eigenvalues converge to harmonic-measure averages for typical frozen
spectra.

\begin{theorem}[Harmonic-measure limit for quenched growth]
	\label{thm:quenched}
	Let $Q$ satisfy the preceding regular analytic assumptions, and let
	$M=M_N$ with
	\begin{equation}\label{eq:mesoscopic-regime}
		\frac{M_N}{\log N}\longrightarrow\infty,
		\qquad \frac{M_N}{N}\longrightarrow0.
	\end{equation}
	By $Y=\{y_1,\ldots,y_N\}$ we denote the eigenvalues of the rank-$N$
	normal-matrix ensemble with weight $e^{-2NQ}dA$. With $Y$ held
	fixed, let $Z=\{z_1,\ldots,z_M\}$ be the $M$ added eigenvalues of
	quenched growth, that is, the $\beta=2$ Coulomb gas in the field of the
	frozen eigenvalues,
	\[
		\Prob(dZ\mid Y)\ \propto\
		\prod_{a<b}|z_a-z_b|^2
		\prod_{a=1}^{M}\Bigl(\prod_{i=1}^{N}|z_a-y_i|^2\Bigr)
		e^{-2NQ(z_a)}\,dA(z_a),
	\]
	as in \eqref{eq:conditional-density}.  Then the empirical measure
	$\widehat\mu_Z=M^{-1}\sum_{a=1}^{M}\delta_{z_a}$ converges weakly to
	$\omega_D$, conditionally on $Y$, outside a set of frozen spectra whose
	$\Prob_N^{\mathrm{old}}$-probability tends to zero.  More precisely, for
	every bounded continuous $f\colon\C\to\mathbb R$ and every
	$\eta,\varepsilon>0$,
	\begin{equation}\label{eq:quenched-convergence}
		\Prob_N^{\mathrm{old}}\biggl\{
		\Prob\biggl(
		\Bigl|\frac1M\sum_{a=1}^{M}f(z_a)
		-\int_{\partial D}f\,d\omega_D\Bigr|>\eta
		\ \Bigm|\ Y\biggr)>\varepsilon\biggr\}
		\longrightarrow0.
	\end{equation}
	In particular, the conditional one-point density of the added
	eigenvalues, divided by $M$, converges weakly to $\omega_D$ in
	$\Prob_N^{\mathrm{old}}$-probability, and so does its average over the
	frozen spectrum.
\end{theorem}

For each frozen spectrum $Y$, \eqref{eq:quenched-convergence} measures
the probability that a new batch deviates from harmonic measure.
The assertion is that this conditional probability tends to zero
in probability when $Y$ is sampled from the old ensemble.
In terms of the empirical measure, this means that, for every open
weak neighborhood $\mathcal U$ of $\omega_D$,
\begin{equation}\label{eq:weak-neighborhood-form}
	\Prob(\widehat\mu_Z\notin\mathcal U\mid Y)\longrightarrow0
	\qquad\text{in }\Prob_N^{\mathrm{old}}\text{-probability}.
\end{equation}
Indeed, one can choose finitely many bounds on test-function averages
that together ensure $\widehat\mu_Z\in\mathcal U$.  Applying
\eqref{eq:quenched-convergence} to each bound and adding the failure
probabilities gives \eqref{eq:weak-neighborhood-form}.
Conversely, each individual test-function bound defines a weak
neighborhood of $\omega_D$.

Averaging the conditional deviation probability over $Y$ gives
convergence under the joint sampling of the old spectrum and the new
batch.  Thus, for every bounded continuous $f$,
$M^{-1}\sum_a f(z_a)$ converges in probability to
$\int f\,d\omega_D$.

\begin{remark}[Coulomb-gas formulation]
	Conditionally on $Y$, the added eigenvalues are $M$ unit charges of a
	two-dimensional $\beta=2$ Coulomb gas in the external potential
	$2NQ(z)-2\sum_{i}\log|z-y_i|$.  In the mean-field approximation this
	potential is constant on $D$, since $Q-U_1=C_1$ there.  The new charges
	therefore see a conductor with equilibrium charge distribution
	$\omega_D$.  Theorem~\ref{thm:quenched} proves this harmonic-measure law
	for $\log N\ll M\ll N$ by controlling the fluctuations of $q_Y$.
\end{remark}

Together with Proposition~\ref{prop:band-harmonic-measure}, the theorem
shows that the two growth rules have the same limiting spatial
distribution when $\log N\ll M\ll N$.

\begin{corollary}[Comparison with equilibrium growth]
	\label{cor:comparison}
	Under the assumptions of Theorem~\ref{thm:quenched}, both the orbital-band
	process \eqref{eq:equilibrium-band-process} and one quenched mesoscopic
	batch have limiting spatial distribution $\omega_D$.  For the orbital
	band, the empirical measure converges in probability.  For the quenched
	batch, convergence holds conditionally as in
	\eqref{eq:quenched-convergence}.  By \eqref{eq:rule-a-drift}, this common
	limit is the derivative of the equilibrium measure at filling $t=1$.
\end{corollary}
\begin{proof}
	Proposition~\ref{prop:band-harmonic-measure} gives the finite-band
	limit, since the theorem's regime implies $M\to\infty$ and $M/N\to0$.
	Theorem~\ref{thm:quenched} gives the quenched limit, and
	\eqref{eq:rule-a-drift} identifies the continuum derivative.
\end{proof}

\section{Proof of Theorem~\ref{thm:quenched}}\label{sec:proof}
The proof rests on an energy comparison.  Harmonic measure maximizes
the logarithmic interaction energy on the droplet, so batches whose
spatial distribution stays away from it are suppressed on the $M^2$
scale.  For typical frozen spectra, our bound on the frozen field
introduces an error of order $M\log N$, which is smaller when
$M\gg\log N$.  A separate confinement estimate controls the probability
that new eigenvalues leave a small neighborhood of the droplet.

By Proposition~\ref{prop:polynomial}, the conditional probability that
the added eigenvalues fail to concentrate near $\omega_D$ is a ratio of
partition functions: for an open weak neighborhood $\mathcal U$ of
$\omega_D$ and $\widehat\mu_Z=M^{-1}\sum_a\delta_{z_a}$,
\begin{equation}\label{eq:partition-ratio}
	\Prob(\widehat\mu_Z\notin\mathcal U\mid Y)
	=\frac{1}{M!\det G_Y}
	\int_{\{\widehat\mu_Z\notin\mathcal U\}}
	|V(Z)|^2\prod_{a=1}^{M}|q_Y(z_a)|^2\,d\lambda_N(z_a).
\end{equation}
Orthogonal-polynomial norms bound the denominator below for every $Y$
(Section~\ref{sec:proof-normalization}).  For typical frozen spectra, we
bound $|q_Y|^2e^{-2NQ}$ near $D$ (Section~\ref{sec:proof-envelope}).
The logarithmic-energy gap outside $\mathcal U$ bounds the Vandermonde determinant
(Section~\ref{sec:proof-energy}).
A projection-kernel estimate controls escape from this neighborhood
(Section~\ref{sec:proof-confinement}).
We use the leading-coefficient asymptotics, local submean inequality,
exterior bound, and uniform $L^2$ orbital expansion of
\cite{HedenmalmWennman}.  Constants may depend on $Q$ and the fixed collar,
but not on $N$, $M$, or $Y$.

\subsection{The orthogonal-polynomial normalization}
\label{sec:proof-normalization}
For a filling $t$ near $1$,
let $\mu_t$ be the equilibrium eigenvalue measure of mass $t$, supported
on $D_t$, and define its logarithmic potential by
\[
	U_t(z)=\int\log|z-w|\,d\mu_t(w).
\]
The weighted equilibrium condition gives a real constant $C_t$ such that
\[
	Q-U_t=C_t\quad\hbox{on }D_t,
	\qquad
	Q-U_t\ge C_t\quad\hbox{on }\C.
\]
Thus, the obstacle solution
\begin{equation}\label{eq:obstacle}
	\widehat Q_t=U_t+C_t
\end{equation}
coincides with $Q$ on $D_t$, lies below $Q$ on $\C$, and is harmonic
on $\C\setminus D_t$.

Let $\varphi_t$ be the normalized exterior conformal map of $D_t$ with logarithmic capacity $r_t=\capc(D_t)$.
We introduce a bounded holomorphic function $\mathcal Q_t$  on $\C\setminus D_t$,
whose real part has boundary values $Q|_{\partial D_t}$ and whose
imaginary part vanishes at infinity. Below, we put $A_t=\mathcal Q_t(\infty)$.
Because of uniqueness of the exterior harmonic extension,
one obtains
\[
	\widehat Q_t=t\log|\varphi_t|+\operatorname{Re}\mathcal Q_t
	\qquad\hbox{on }\C\setminus D_t.
\]
Since $U_t(z)=t\log|z|+o(1)$ and
$\varphi_t(z)=z/r_t+O(1)$ at infinity, comparison with
\eqref{eq:obstacle} yields the relation $C_t=A_t-t\log r_t$.

The pointwise asymptotic of
\cite[Theorem~1.3 and Remark~1.4(b)]{HedenmalmWennman},
evaluated at infinity, gives the leading coefficient of the normalized
degree-$n$ polynomial~\eqref{eq:orthogonal-polynomial} as
$N^{1/4}r_t^{-n-1/2}e^{NA_t}
	\mathcal B_{0,t}(\infty)(1+O(N^{-1}))$ for $t=n/N$ near $1$.
Using \cite[Theorem~1.5]{HedenmalmWennman} we readily determine
$\mathcal B_{0,t}=\pi^{-1/4}e^{H_{Q,t}}$, where
$\operatorname{Re}H_{Q,t}=\tfrac14\log\Delta Q$ on $\partial D_t$.
Note that strict subharmonicity keeps $\mathcal B_{0,t}(\infty)$ bounded above and
away from zero, uniformly over this compact family of boundaries.
Taking the reciprocal square of the leading coefficient gives (with
$C_t=A_t-t\log r_t$),
\begin{equation}\label{eq:norm-asymptotic}
	\log h_{n,N}=-2NC_t-\tfrac12\log N+O(1)
\end{equation}
uniformly for $t=n/N$ near $1$.

The equilibrium growth law \eqref{eq:rule-a-drift}, valid on both sides
of $t=1$ in this regular flow \cite[Section~2.3]{HedenmalmWennman},
also determines the derivative of $C_t$.
Let us fix a point $x$ that remains inside $D_t$ for all $t$ near $1$.
Since $d\mu_t/dt=\omega_{D_t}$, differentiating the logarithmic
potential gives
\[
	\frac{d}{dt}U_t(x)
	=\int_{\partial D_t}\log|x-w|\,d\omega_{D_t}(w)
	=\log r_t,
\]
where the last equality uses the constant logarithmic potential of
harmonic measure inside the droplet.
Because $Q$ is independent of $t$, differentiating
$Q(x)=U_t(x)+C_t$ now yields $0=\log r_t+C_t'$. Consequently,
\begin{equation}\label{eq:norm-derivative}
	\frac{d}{dt}(-2C_t)=2\log r_t.
\end{equation}
As $M/N\to0$, the smooth regular flow and
\eqref{eq:norm-asymptotic}--\eqref{eq:norm-derivative} give, uniformly
for $0\le k<M$,
\begin{equation}\label{eq:capacity-norms}
	\log\frac{h_{N+k,N}}{h_{N,N}}
	=2k\log r+O\!\left(\frac{k^2}{N}+1\right).
\end{equation}
An $o(M)$ remainder uniform in $k<M$ would suffice below.

Consider the Gram matrix $G_Y$ of
$q_Y,q_Yz,\ldots,q_Yz^{M-1}$ in Proposition~\ref{prop:polynomial}. Let $g_k$ be the polynomial obtained from $q_Yz^k$ by subtracting
its projections onto the polynomials already orthogonalized.
Since these subtractions involve only lower-degree terms,
$g_k$ remains monic of degree $N+k$.

Subtracting earlier basis elements leaves the Gram determinant
unchanged.  The resulting polynomials are mutually orthogonal,
so their Gram matrix is diagonal and
\[
	\det G_Y=\prod_{k=0}^{M-1}\|g_k\|_{L^2(\lambda_N)}^2.
\]
Since $g_k$ and $p_{N+k,N}$ are both monic of degree $N+k$,
their difference has lower degree and is therefore orthogonal to
$p_{N+k,N}$.  Thus, $\|g_k\|_{L^2(\lambda_N)}^2
	=h_{N+k,N}
	+\|g_k-p_{N+k,N}\|_{L^2(\lambda_N)}^2
	\ge h_{N+k,N}$.
Each factor in the product for $\det G_Y$ is therefore at least
$h_{N+k,N}$.
Therefore, for every distinct $Y$, we have a lower bound on $\det G_Y$:
\begin{equation}\label{eq:gram-lower}
	\det G_Y\ge\prod_{k=0}^{M-1}h_{N+k,N},\qquad
	\log\frac{\det G_Y}{h_{N,N}^{M}}
	\ge M(M-1)\log r+o(M^2).
\end{equation}

The second inequality follows by summing \eqref{eq:capacity-norms}, which
is uniform in $Y$.  Since $\log M!=O(M\log M)=o(M^2)$, the factorial does not affect the leading $M^2$ exponent.

\subsection{A one-sided bound on the frozen polynomial}
\label{sec:proof-envelope}
We compute the mean squared modulus of the frozen polynomial
$q_Y(z)=\prod_i(z-y_i)$.
The identity
$|V(Y)|^2\prod_i|z-y_i|^2=|V(Y\cup\{z\})|^2$
shows that the integral over $Y$ is the same integral that defines
the rank-$(N+1)$ one-point intensity, up to normalization.
The one-point intensity counts all $N+1$ eigenvalues, so comparison
with the old ensemble introduces the normalization factor
\[
	\frac{\mathcal Z_{N+1,N}}
	{(N+1)\mathcal Z_{N,N}}=h_{N,N},
\]
by \eqref{eq:tau-normalizer}.  Therefore
\begin{equation}\label{eq:expected-field}
	\E_Y|q_Y(z)|^2=h_{N,N}\,K_{N+1}(z,z),
	\qquad
	\E_Y\bigl[\rho_Y(z)\bigr]=R_{N+1,N}(z),
	\qquad
	\rho_Y=\frac{|q_Y|^2e^{-2NQ}}{h_{N,N}}.
\end{equation}
Here, the first identity is the $\beta=2$ characteristic-polynomial formula
$\E|\det(z-X)|^2=h_{N,N}K_{N+1}(z,z)$.  The second identifies the averaged
frozen field with the next-rank eigenvalue density, which is $O(N)$ on the
droplet.  Integrating \eqref{eq:expected-field} over
$\C$ gives
\begin{equation}\label{eq:expected-q}
	\E_Y\|q_Y\|_{L^2(\lambda_N)}^2=(N+1)h_{N,N}.
\end{equation}
For any fixed $B>1$, let $\mathcal G_N$ be the set of frozen spectra
satisfying $\|q_Y\|_{L^2(\lambda_N)}^2\le N^B h_{N,N}$.
Markov's inequality bounds the probability of exceeding this threshold
by the mean squared norm divided by the threshold.  Using
\eqref{eq:expected-q}, we obtain
\begin{equation}\label{eq:good-event}
	\Prob_Y(\mathcal G_N^{\mathsf c})
	\le (N+1)N^{-B}\longrightarrow0.
\end{equation}
For these spectra, the norm bound also controls the frozen field
pointwise near the droplet.  Fix a compact collar $K$ of $D$.
The local weighted submean inequality
\cite[Lemma~2.2]{HedenmalmWennman} gives
\begin{equation}\label{eq:random-envelope}
	\sup_{z\in K}\rho_Y(z)
	\le C_K N^{B+1},\qquad Y\in\mathcal G_N.
\end{equation}

\subsection{The logarithmic-energy gap}
\label{sec:proof-energy}
For a probability measure $\nu$ on a compact set put
\begin{equation}\label{eq:log-energy}
	\mathcal E(\nu)=\iint\log|z-w|\,d\nu(z)d\nu(w).
\end{equation}
Among all probability measures supported in $D$, exterior harmonic
measure $\omega_D$ uniquely maximizes the logarithmic energy:
$\mathcal E(\omega_D)=\log r$.
It is therefore called the logarithmic equilibrium measure of $D$.
Although the measures in this variational problem may occupy the
whole droplet, the maximizing measure is supported on its boundary
\cite{SaffTotik}.  Given an open weak neighborhood
$\mathcal U$ of $\omega_D$, choose a small closed collar
$K_\delta=\{z:\dist(z,D)\le\delta\}\Subset\operatorname{int}K$.
The bound \eqref{eq:random-envelope} applies on $K_\delta$.
Compactness, upper semicontinuity of $\mathcal E$, and uniqueness of
$\omega_D$ imply that
one can choose $\delta,c_{\mathcal U}>0$ such that
\begin{equation}\label{eq:energy-gap}
	\sup_{\substack{\nu\in\mathcal P(K_\delta)\\
			\nu\notin\mathcal U}}
	\mathcal E(\nu)\le\log r-2c_{\mathcal U}.
\end{equation}
If this failed for every shrinking collar, a sequence of near-maximizers
would converge to a measure supported on $D$, outside $\mathcal U$,
with energy at least $\log r$, contradicting uniqueness.

The Vandermonde determinant involves the logarithmic interaction
between distinct points.  To compare it with the energy of a measure,
we first remove the singularity at coincident points by setting
\[
	\ell_T(z,w)=\max\{\log|z-w|,-T\}.
\]
This replacement can only increase each pair interaction.  The double
integral against $\widehat\mu_Z$ also includes $M$ diagonal terms,
each equal to $-T$.  Correcting for these terms gives
\[
	M^{-2}\log|V(Z)|^2
	\le\iint\ell_T(z,w)\,d\widehat\mu_Z(z)d\widehat\mu_Z(w)
	+T/M.
\]

For fixed $T$, the truncated interaction is continuous, so its energy
depends continuously on the measure.  As $T$ increases, these energies
decrease to $\mathcal E$.  Compactness of $S=\{\nu\in\mathcal P(K_\delta):\nu\notin\mathcal U\}$
ensures that their maxima over $S$ decrease to
$\sup_S\mathcal E$.  By \eqref{eq:energy-gap}, we can therefore
choose a fixed $T$ large enough that
\[
	\iint\ell_T(z,w)\,d\nu(z)d\nu(w)
	\le \log r-c_{\mathcal U}
	\qquad(\nu\in S).
\]
Since $T/M\to0$, every configuration in $K_\delta^M$ with
$\widehat\mu_Z\notin\mathcal U$ satisfies
\begin{equation}\label{eq:hard-wall}
	|V(Z)|^2
	\le\exp\{M^2(\log r-c_{\mathcal U}+o(1))\}.
\end{equation}
The bound is uniform over these configurations.  Since $K_\delta$
has fixed finite area, integrating this bound over $K_\delta^M$
adds only $O(M)$ to the logarithm.

\subsection{Confinement near the droplet}
\label{sec:proof-confinement}

We now bound the probability that any added eigenvalue lies outside
$K_\delta$.  This estimate holds uniformly over the frozen spectra.
After fixing $K_\delta$, we choose $t_\ast>1$ close enough to $1$ such that
$D_{t_\ast}\Subset\operatorname{int}K_\delta$ and
$t_\ast<\liminf_{|z|\to\infty}Q(z)/\log|z|$.
Since $M/N\to0$,
every polynomial in $F$ has degree less than $Nt_\ast$ for large $N$.
The weighted submean inequality on $\partial D_{t_\ast}$, followed by
the exterior maximum principle for $\widehat Q_{t_\ast}$ as in
\cite[Proposition~2.3]{HedenmalmWennman}, gives for
every $L^2(\lambda_N)$-normalized $p\in F$,
\begin{equation}\label{eq:exterior-bound}
	|p(z)|^2e^{-2NQ(z)}
	\le CN e^{-2N[Q(z)-\widehat Q_{t_\ast}(z)]},
	\qquad z\notin D_{t_\ast}.
\end{equation}
Let us introduce $W=Q-\widehat Q_{t_\ast}$. Then, the equality of the droplet and the
full coincidence set gives $W>0$ off $D_{t_\ast}$.  Since
$D_{t_\ast}\Subset\operatorname{int}K_\delta$ and
$W\ge\zeta\log|z|-O(1)$ at infinity for some $\zeta>0$,
continuity gives $W\ge b>0$ on $K_\delta^{\mathsf c}$.
Choose a fixed integer $N_0>2/\zeta$.  Then
$\int_{K_\delta^{\mathsf c}}e^{-N_0W}dA<\infty$, and for
$N\ge N_0$,
\[
	\int_{K_\delta^{\mathsf c}}e^{-2NW}dA
	\le e^{-Nb}\int_{K_\delta^{\mathsf c}}e^{-N_0W}dA
	\le C e^{-Nb}.
\]
Summing
\eqref{eq:exterior-bound} over an orthonormal basis of $F$ gives
\begin{equation}\label{eq:confinement}
	\Tr(\mathbf1_{K_\delta^{\mathsf c}}\Pi_F)
	\le C N(N+M)e^{-Nb}=o(1)
\end{equation}
for some $b>0$.  The left side is
$\int_{K_\delta^{\mathsf c}}R_{N+M,N}\,dA$, the expected number of
eigenvalues of the rank-$(N+M)$ ensemble outside the collar.  The operator
inequality $0\le\Pi_Y^\perp\le\Pi_F$ bounds the one-point
density of the added eigenvalues by $R_{N+M,N}$.  Their conditional
expected number outside $K_\delta$ therefore obeys the same bound for
every distinct $Y$.  Thus $\sup_Y\Prob(Z\not\subset K_\delta\mid Y)=o(1)$,
where the supremum is over frozen spectra for which the completion is
defined.

\subsection{Combining the estimates}
\label{sec:proof-combine}
Dividing numerator and denominator of \eqref{eq:partition-ratio} by
$h_{N,N}^{M}$ writes the conditional Boltzmann weight as
$|V(Z)|^2\prod_a\rho_Y(z_a)dA(z_a)$ and the conditional partition
function as $M!\det(G_Y/h_{N,N})$.  On $\mathcal G_N$, the numerator
restricted to $Z\subset K_\delta$, $\widehat\mu_Z\notin\mathcal U$ is
bounded by \eqref{eq:random-envelope} and \eqref{eq:hard-wall}.  The
denominator is controlled by \eqref{eq:gram-lower}, while the part of the
numerator with $Z\not\subset K_\delta$ is the confinement error
\eqref{eq:confinement}.  Uniformly on $\mathcal G_N$, we have the following estimate
\begin{equation}\label{eq:final-bound}
	\Prob(\widehat\mu_Z\notin\mathcal U\mid Y)
	\le\exp\{-c_{\mathcal U}M^2+O(M\log N)+o(M^2)\}+O(N^2e^{-bN}).
\end{equation}
Both terms tend to zero, because $M/\log N\to\infty$ controls the first, and
\eqref{eq:confinement} controls the second.  Since
$\Prob_Y(\mathcal G_N)\to1$, this proves \eqref{eq:weak-neighborhood-form},
and hence \eqref{eq:quenched-convergence}.  The one-point conclusions follow
by testing against bounded continuous functions and the triangle
inequality. \qed

\begin{remark}[Quantitative form]
	For fixed $B>1$ and all large $N$, the estimates hold on an event of
	$\Prob_N^{\mathrm{old}}$-probability at least $1-(N+1)N^{-B}$:
	\[
		\Prob(\widehat\mu_Z\notin\mathcal U\mid Y)
		\le e^{-c_{\mathcal U}M^2/2}+CN^2e^{-bN}.
	\]
	Taking $B=B_N\to\infty$ with $B_N=o(M/\log N)$ makes the exceptional
	set of frozen spectra superpolynomially small.
\end{remark}

\section{A three-fold droplet}\label{sec:examples}

We illustrate the two growth rules for a droplet with three-fold
symmetry.  Its explicit equilibrium flow and harmonic measure provide
a reference for comparing the spatial distribution of the added
eigenvalues and their angular correlations.
\subsection{Droplet and boundary coordinate}
The three-fold map
\begin{equation}\label{eq:threefold-map}
	f(w)=rw+\frac{a}{w^2},\qquad |w|>1,
\end{equation}
with $2a<r$ and $r^2-2a^2=1$, maps the unit circle to the boundary
of the unit-filling droplet $D=D_1$ for the potential
\begin{equation}\label{eq:threefold-potential}
	Q(z)=\tfrac12|z|^2-\tau\operatorname{Re}z^3,
	\qquad \tau=\frac{a}{3r^2}.
\end{equation}
Since the cubic potential is not confining, we restrict the eigenvalues
to $|z|\le R$ by setting $Q_R=Q$ inside the disk and $Q_R=+\infty$
outside.  We use $a=0.2$ and $R=2$.  Such cutoffs are standard for
polynomial normal-matrix potentials \cite{ElbauFelder}, but the sharp
wall falls outside the smooth-potential assumptions of
Theorem~\ref{thm:quenched}.

The exterior map $\varphi$ in \eqref{eq:conformal-map} satisfies
$\varphi^{-1}=f$.  Harmonic measure is uniform in the conformal angle
$\theta$, though generally nonuniform in arclength.  We assign an angle
to each sampled point by projecting it onto the closest boundary point:
\begin{equation}\label{eq:closest-boundary-projection}
	\varphi^{-1}(e^{i\Theta_D(z)})\in
	\operatorname*{arg\,min}_{\xi\in\partial D}|z-\xi|.
\end{equation}
If the closest point is not unique, we choose the one with the smallest
angle in $[0,2\pi)$.  Let $u_D(z)$ denote the signed distance to
$\partial D$, negative inside the droplet and positive outside.

Since $\Theta_D$ is continuous near the smooth boundary, it transfers
the spatial limit in Theorem~\ref{thm:quenched} to a uniform angular
limit, for admissible weights in the regime of the theorem.
Write $\widehat\nu_Z=(\Theta_D)_\#\widehat\mu_Z$ for the angular
distribution of the sampled batch and
$\nu_Y=(\Theta_D)_\#\E_{N,M}^{Y}[\widehat\mu_Z]$ for its conditional mean.
For a continuous $2\pi$-periodic function $g_{\mathrm{test}}$, let
\[
	c_{g_{\mathrm{test}}}=\int_0^{2\pi}g_{\mathrm{test}}(\theta)\,\frac{d\theta}{2\pi}.
\]
Then, for any fixed $\eta,\varepsilon>0$, both deviation probabilities
tend to zero:
\begin{equation}\label{eq:projected-harmonic-limit}
	\Prob_N^{\mathrm{old}}\!\left\{Y:
	\Prob\!\left(\left|\int g_{\mathrm{test}}\,d\widehat\nu_Z-c_{g_{\mathrm{test}}}\right|>\eta
	\,\middle|\,Y\right)>\varepsilon\right\},
	\quad
	\Prob_N^{\mathrm{old}}\!\left\{Y:
	\left|\int g_{\mathrm{test}}\,d\nu_Y-c_{g_{\mathrm{test}}}\right|>\eta\right\}
	\longrightarrow0.
\end{equation}
We use this uniform angular distribution as a reference for the
hard-wall numerical example below.

\subsection{One-point convergence}
We measure departures from harmonic measure through polynomial moments.
Any polynomial can be centered by subtracting its harmonic-measure mean.
We use the Faber polynomials $F_n$, whose normalization is fixed by the
exterior map through $F_n(f(w))=w^n+O(w^{-1})$.
On the boundary, this gives the angular mode $e^{in\theta}$ together
with negative Fourier modes, all with zero uniform average.
Thus $\int F_n\,d\omega_D=0$ for $n\ge1$.
Removing the constant term of a polynomial would not generally give
this property: for example, $F_3(z)=(z/r)^3-3a/r$ needs a nonzero
constant term to have zero harmonic-measure mean.
We use the first six moments to define
\begin{equation}\label{eq:faber-discrepancy}
	\mathcal D_6(Y;N,M)^2
	=\sum_{n=1}^{6}\frac{1}{n^2}
	\left|\int F_n\,d\,\E_{N,M}^{Y}[\widehat\mu_Z]\right|^2 .
\end{equation}
These moments are evaluated at the original points, before boundary
projection.  They are finite in the bounded hard-wall region.
For general admissible weights, convergence of these unbounded
polynomial observables requires additional moment estimates beyond
Theorem~\ref{thm:quenched}.

For the continuum equilibrium comparison, we average harmonic measure
along the growing droplets $D_t$, $1\le t\le1+\gamma$, where
$\gamma=M/N$, keeping the Faber polynomials of the old droplet fixed.
Three-fold symmetry leaves only the third and sixth moments.
Set $\alpha=a/r$ and $x(t)=3\tau r_t^4/r^3$, where the conformal radius
$r_t$ follows the regular branch of
$r_t^2-18\tau^2r_t^4=t$ with $r_1=r$.
An overbar denotes the average over $1\le t\le1+\gamma$.
Evaluating the two nonzero moments gives
\begin{equation}\label{eq:equilibrium-faber-discrepancy}
	\mathcal D_{6,\mathrm{eq}}(\gamma)^2
	=|\overline x-\alpha|^2
	+\frac{|15\,\overline{x^2}-18\alpha\overline x+3\alpha^2|^2}{36}.
\end{equation}
This expression gives the continuum equilibrium curves in
Figure~\ref{fig:threefold-harmonic}(c).  For $M/N$ small, the continuum equilibrium discrepancy has the expansion
\begin{equation}\label{eq:equilibrium-faber-small-band}
	\mathcal D_{6,\mathrm{eq}}(\gamma)
	=\frac{|\alpha|\sqrt{1+4\alpha^2}}
	{r^2(1-4\alpha^2)}\,\gamma+O(\gamma^2).
\end{equation}
For $a=0.2$, the coefficient is approximately $0.2241$, and
the continuum equilibrium discrepancy grows linearly with $M/N$.

\begin{lemma}[Conditional Faber fluctuations]\label{lem:faber-fluctuations}
	Let $\Pi_Y^\perp$ be the conditional projection, with range
	$q_Y\Poly_{<M}$ in $L^2(\lambda_N)$, and suppose the sampling region is
	contained in $\{|z|\le R\}$.  For every polynomial $p$ of degree
	$d<M$,
	\begin{equation}\label{eq:polynomial-variance-bound}
		\operatorname{Var}\!\left[
			\frac1M\sum_{z\in Z}p(z)\,\middle|\,Y
			\right]
		\le \frac{d}{M^2}\|p\|_{L^\infty(|z|\le R)}^2.
	\end{equation}
	Consequently, for $M>6$,
	\begin{equation}\label{eq:faber-fluctuation-bound}
		\E\!\left[
			\sum_{n=1}^6\frac1{n^2}
			\left|\int F_n\,d\left(
			\widehat\mu_Z-\E^Y_{N,M}[\widehat\mu_Z]
			\right)\right|^2\,\middle|\,Y
			\right]
		\le\frac1{M^2}\sum_{n=1}^6
		\frac{\|F_n\|_{L^\infty(|z|\le R)}^2}{n}.
	\end{equation}
\end{lemma}
\begin{proof}
	Choose an orthonormal basis $\psi_1,\ldots,\psi_M$ of
	$q_Y\Poly_{<M}$ so that the first $M-d$ elements span
	$q_Y\Poly_{<M-d}$.  Multiplying any of these first elements by
	$p$ gives a polynomial still in $q_Y\Poly_{<M}$, so
	\[
		(1-\Pi_Y^\perp)(p\psi_j)=0,\qquad j\le M-d.
	\]
	Only the last $d$ basis elements therefore contribute to the
	conditional variance.  The variance formula for a projection
	process gives
	\[
		\operatorname{Var}\!\left[
			\frac1M\sum_{z\in Z}p(z)\,\middle|\,Y
			\right]
		=\frac1{M^2}\sum_{j=M-d+1}^{M}
		\|(1-\Pi_Y^\perp)(p\psi_j)\|_{L^2(\lambda_N)}^2.
	\]
	An orthogonal projection cannot increase the norm, and
	$\|\psi_j\|_{L^2(\lambda_N)}=1$.  Each summand is thus at most
	$\|p\|_{L^\infty(|z|\le R)}^2$, proving
	\eqref{eq:polynomial-variance-bound}.
	Applying this bound to $F_n$, which has degree $n$, and summing
	with weights $1/n^2$ gives \eqref{eq:faber-fluctuation-bound}.
\end{proof}

The lemma controls the variation between batches sampled from the same
frozen spectrum.  We now examine the conditional mean numerically,
using the projected density and the Faber discrepancy.

In Figure~\ref{fig:threefold-harmonic}(b), the projected quenched density
is close to the harmonic profile in \eqref{eq:projected-harmonic-limit}.
Panel (c) shows a decreasing Faber discrepancy over $30\le M\le150$, while
the continuum equilibrium curves follow
\eqref{eq:equilibrium-faber-discrepancy}. Panel (d) measures the fraction of points farther than $0.1$ from the
old boundary.  In the sampled batches, all such points lie inside the
droplet.  Their mean count stays between $6.5$ and $7.9$, while their
mean fraction falls from about $0.2$ at $M=30$ to $0.05$
at $M=150$.  The falling fraction therefore mainly reflects the
increase in batch size. The figures use 32 frozen spectra at $N=5000$, 160 at $N=10^4$, and
30 conditional batches per spectrum for each $M\in\{30,60,100,150\}$.
Bootstrap intervals describe sampling uncertainty.  Grid refinement
changes the first 15 conditional angular means by at most
$8.4\times10^{-5}$ at $(N,M)=(10^4,60)$.

\begin{figure}[t]
	\centering
	\includegraphics[width=\textwidth]{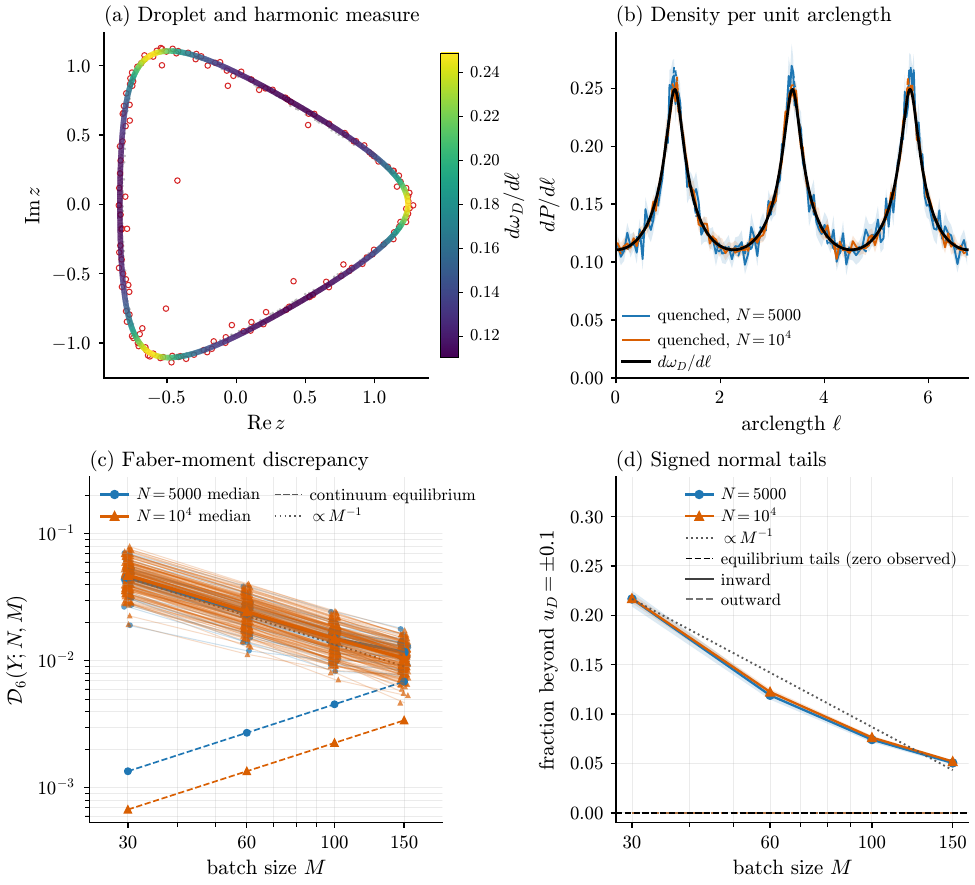}
	\caption{One-point observables on the three-fold droplet at $N=5000$
		and $N=10^4$.  (a) Boundary coloured by harmonic-measure density,
		with one $M=150$ quenched batch.  (b) Projected density per unit
		arclength at $M=150$, compared with harmonic measure.
		(c) Faber discrepancy \eqref{eq:faber-discrepancy}:
		thin lines show individual frozen spectra, thick solid lines
		show medians, dashed lines show continuum equilibrium values,
		and the dotted line is an $M^{-1}$ guide.
		(d) Fractions of points farther than $0.1$
		from the boundary.}
	\label{fig:threefold-harmonic}
\end{figure}

We now compare the conditional mean density with the equilibrium band
density at finite $N$.  Each is given by the diagonal of its projection
kernel.  Since
$\Pi_Y^\perp=\Pi_F-P_Y$ and
$\Pi^{\mathrm{band}}_{N,M}=\Pi_F-\Pi_H$,
the common contribution $\Pi_F$ cancels, giving
\begin{equation}\label{eq:vacancy-decomposition}
	\E_{N,M}^{Y}[M\widehat\mu_Z]-M\mu^{\mathrm{eq}}_{N,M}
	=\operatorname{diag}(\Pi_H-P_Y)e^{-2NQ}\,dA.
\end{equation}
Here $\operatorname{diag}T$ denotes the diagonal of the integral kernel
of $T$ relative to $d\lambda_N$.  Both densities count $M$ added eigenvalues,
so their difference has total mass zero.

To compare the Faber moments, we integrate this difference against
polynomials.  The following proposition expresses these integrals
through fluctuations of the frozen spectrum and a finite-rank remainder.
Dividing by $M$ gives the correction to the normalized mean moments
used above.  This identifies a factor $1/M$, although an $M^{-1}$ rate
still requires bounds on both the frozen-spectrum fluctuations and
the remainder.

\begin{proposition}[Finite-rank moment identity]\label{prop:screening}
	Let $Y$ be distinct, and let $\ell_1,\ldots,\ell_N\in\operatorname{Ran}P_Y$
	satisfy $\ell_i(y_j)=\delta_{ij}$.  Let $p$ be a polynomial of degree $d$
	and assume $\Poly_{<N+M+d}\subset L^2(\lambda_N)$, as holds for the
	bounded hard-wall sampling region.  Then
	\begin{equation}\label{eq:screening}
		M\!\int p\,d\bigl(\E^Y_{N,M}[\widehat\mu_Z]-\mu^{\mathrm{eq}}_{N,M}\bigr)
		=-\Bigl(\sum_{i=1}^Np(y_i)-\E_N^{\mathrm{old}}\sum_{i=1}^Np(y_i)\Bigr)
		+\mathcal R_p(Y),
	\end{equation}
	where
	\[
		\mathcal R_p(Y)=\sum_{i=1}^N\bigl[(1-\Pi_F)(p\,\ell_i)\bigr](y_i),
	\]
	and $(1-\Pi_F)(p\,\ell_i)$ lies in the span of
	$\phi_{N+M,N},\ldots,\phi_{N+M+d-1,N}$.
\end{proposition}
\begin{proof}
	A polynomial in $\operatorname{Ran}P_Y$ that vanishes on $Y$ must be zero.
	Evaluation is therefore invertible on this $N$-dimensional space, so the
	interpolation polynomials $\ell_i$ exist.
	For $f\in F$, the polynomial $f-\sum_if(y_i)\ell_i$
	lies in $F$ and vanishes on $Y$, hence is orthogonal to
	$\operatorname{Ran}P_Y$.  Thus $P_Yg=\sum_i\langle g,K_F(\cdot,y_i)\rangle\ell_i$
	for $g\in L^2(\lambda_N)$, and
	\[
		\Tr(pP_Y)=\sum_i\langle p\,\ell_i,K_F(\cdot,y_i)\rangle
		=\sum_i\bigl(\Pi_F(p\,\ell_i)\bigr)(y_i)
		=\sum_ip(y_i)-\mathcal R_p(Y).
	\]
	Since $\Tr(p\,\Pi_H)=\E_N^{\mathrm{old}}\sum_ip(y_i)$, integrating $p$
	against \eqref{eq:vacancy-decomposition} gives \eqref{eq:screening}.
	Finally, $p\,\ell_i\in\Poly_{<N+M+d}$.
\end{proof}
In \eqref{eq:screening}, the correction to the equilibrium mean contains
the negative fluctuation of $\sum_i p(y_i)$ about its expectation and a
remainder from the $d$ polynomial modes above $F$.  An $M^{-1}$ rate for
this discrepancy still requires a bound on $\mathcal R_p(Y)$.

\subsection{Angular correlations}
We compare quenched and equilibrium projected angles at $N=10^4$,
$M=60$, using finite-$M$ CUE and fixed-count independent angles as references.
The CUE curve is motivated by the semiclassical equilibrium law in conformal
coordinates~\cite{AlekseevTMF} and the fixed-$M$ limit of
Proposition~\ref{prop:radial-band-cue}.  Its accuracy for the hard-wall
ensemble is assessed numerically.

We measure the pair correlation $g(s)$ at separation
$s=M\Delta\theta/(2\pi)$, where $\Delta\theta\in[0,\pi]$ is the shorter
angular separation. Thus, $s$ is measured in mean-spacing units.
We count the points on an arc containing $L$ points on average.
The number variance $\Sigma^2(L)$ is the variance of this count,
averaged over all positions of the arc. For the structure factor, we set
\[
	c_n(Z)=\frac1M\sum_{j=1}^M e^{in\Theta_D(z_j)},
	\qquad m_n(Y)=\E[c_n(Z)\mid Y],
\]
and define $S(n)=M\E_{Y,Z}|c_n(Z)|^2$.  Its exact conditional decomposition is
\begin{equation}\label{eq:structure-factor-decomposition}
	S(n)=M\E_Y\operatorname{Var}(c_n\mid Y)
	+M\E_Y|m_n(Y)|^2.
\end{equation}
We compute the conditional mean and variance directly from the
conditional kernel on the numerical grid, then average over the
frozen spectra.
The pair correlation, number variance, and structure factor all depend
on the angular pair distribution averaged over the starting angle. At finite $N$, a nonuniform one-point density
also contributes to $S(n)$ through its squared Fourier coefficients.
For CUE, one has $S_{\mathrm{CUE}}(n)=\min(n,M)/M$
\cite{DiaconisShahshahani}. In Appendix~\ref{app:equilibrium-cue} we give a
separate radial equilibrium-band calculation of finite-$N$ corrections to
CUE.

\begin{figure}[t]
	\centering
	\includegraphics[width=\textwidth]{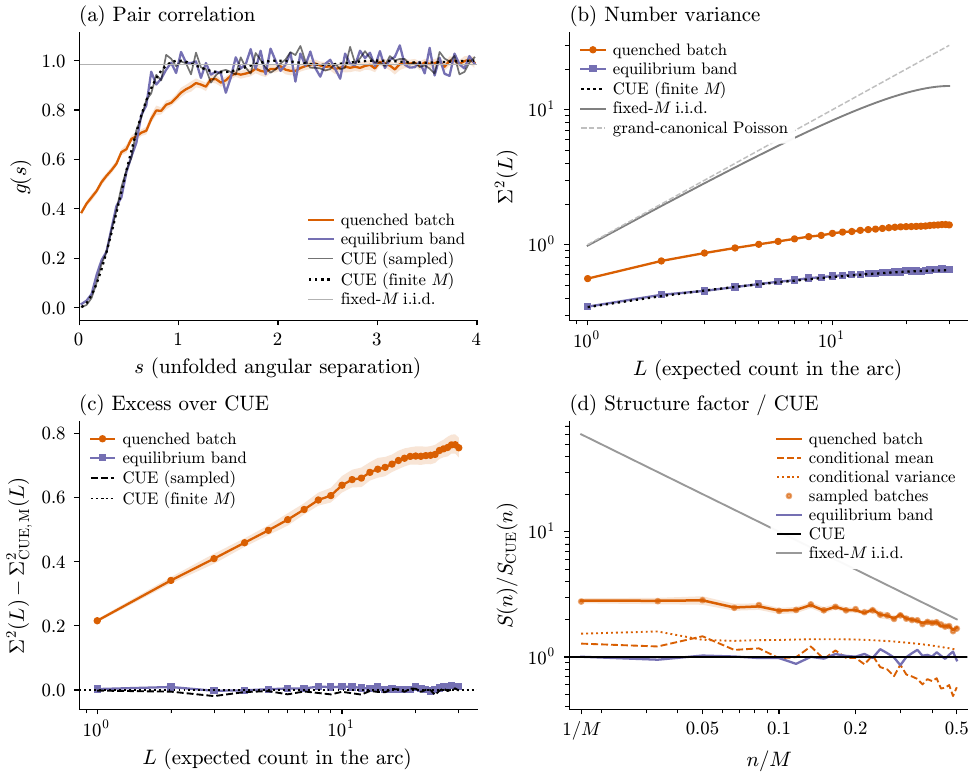}
	\caption{Projected-angle correlations at $N=10^4$, $M=60$.
		(a) Pair correlation.  (b) Absolute number variance.
		(c) Number-variance excess over finite-$M$ CUE.
		(d) Structure factor in CUE units, with the conditional
		mean--variance split \eqref{eq:structure-factor-decomposition}
		and a sampled-batch check.  Shading shows 10th--90th percentile
		bootstrap intervals. Systematic numerical errors are excluded.}
	\label{fig:fluctuations}
\end{figure}

Figure~\ref{fig:fluctuations} shows weaker angular repulsion and larger
fluctuations in the quenched batches than in the equilibrium band.
In panel (a), the pair correlation in the first separation bin,
$0\le s<0.05$, is about $0.39$ for quenched growth and $0.02$ for the band. Panel (b) compares the absolute number variances.  At $L=30$, the
quenched value is $1.40$, compared with $0.645$ for finite-$M$ CUE and
$15$ for fixed-count independent angles.  Panel (c) displays the excess
over CUE separately: the equilibrium band stays close to the CUE
reference, while the quenched variance is larger throughout the measured
range of arc sizes. Panel (d) shows the structure factor and its conditional decomposition.
Averaged over $1\le n\le30$, the quenched structure factor is $2.19$ CUE
units, with contributions $0.87$ from the conditional mean and $1.32$
from the conditional variance; the band gives $1.02$.
The conditional-mean term here concerns projected Fourier modes,
whereas \eqref{eq:screening} concerns polynomial moments.

Quenched batches are still much more rigid than independent angles.
Radial separation could weaken their projected repulsion, but the data
establish neither this mechanism nor a limiting contact value.
To establish hyperuniformity, one must show that the variance of the
number of points in a long arc grows more slowly than its mean count
as the number of points tends to infinity.

\section{Conclusion}
\label{sec:discussion}

The new eigenvalues approach the same boundary distribution whether
the old spectrum is frozen or the growth is described by the equilibrium
band.  For a regular analytic droplet and
\[
	\log N\ll M\ll N,
\]
their spatial distribution converges weakly to exterior harmonic measure
$\omega_D$ on the old boundary.  Theorem~\ref{thm:quenched} gives this
limit conditionally on the frozen spectrum: the probability of exceptional
old configurations tends to zero.  Proposition~\ref{prop:band-harmonic-measure}
gives the corresponding equilibrium-band limit, consistent with the
continuum growth law \eqref{eq:rule-a-drift}.

At finite $N$, the two laws differ.  The equilibrium density increment
comes from the polynomial modes of degrees $N,\ldots,N+M-1$, which form
the band process \eqref{eq:equilibrium-band-process}.  Quenched growth
instead samples from the conditional law \eqref{eq:completion} with the
old eigenvalues held fixed.  Its density contains the factor $|q_Y|^2$, so
the added eigenvalues interact with the actual frozen configuration
(Proposition~\ref{prop:polynomial}).

The hard-wall computation for the three-fold droplet in
Section~\ref{sec:examples} illustrates this comparison.  At $N=5000$ and
$10^4$, the projected quenched density is close to harmonic measure,
and its Faber discrepancy decreases as $M$ grows.  The continuum
equilibrium discrepancy is computed explicitly in
\eqref{eq:equilibrium-faber-discrepancy}.  For batches sampled from the
same frozen spectrum, Lemma~\ref{lem:faber-fluctuations} bounds the
root-mean-square fluctuations of these moments about their conditional
means by a constant times $M^{-1}$.  Proposition~\ref{prop:screening}
expresses the difference between the quenched and equilibrium mean
moments through fluctuations of the old spectrum and a finite-rank
remainder.  Bounds on these terms could turn the identity into
quantitative asymptotics for the conditional mean moments, including
a convergence rate for their difference from the equilibrium values.

The angular correlations also distinguish the two laws.  At
$(N,M)=(10^4,60)$, the equilibrium-band results are close to finite-$M$
CUE.  Quenched batches show weaker angular repulsion and larger number
variance and structure factor, although their counts fluctuate much less
than those of independent angles.  The decomposition
\eqref{eq:structure-factor-decomposition} separates two contributions
to the excess over CUE: the angular Fourier components of the conditional
mean and fluctuations of batches sampled at fixed spectrum.
Both contribute in the computation.

The limiting law of the angular fluctuations remains open.  It would
require control of both the conditional mean and the fluctuations of
batches sampled at fixed old spectrum.

Extending the one-batch theorem to an evolving configuration requires
estimates that remain uniform over successive growth steps.
A connection with particle attachment would further require
quantitative control of the distance of the new eigenvalues from the
boundary.  These estimates would clarify how the correlated batches
studied here relate to stochastic Laplacian growth and
diffusion-limited aggregation.


\section*{Acknowledgments}


The research is supported by PJSC ``Gazprom Neft'', aggr. № \foreignlanguage{russian}{ГПН-26/03000/00502/Р}

\section*{Statements and Declarations}

\textbf{Competing interests.}
The author declares no competing interests.

\textbf{Data availability.}
The numerical data generated in this study and the scripts used to produce
the figures are available from the author upon reasonable request.

\appendix
\clearpage
\section{Equilibrium-band angular limits and a radial finite-size correction}
\label{app:equilibrium-cue}

\subsection{Fixed-batch CUE limit}
\begin{proof}[Proof of Proposition~\ref{prop:radial-band-cue}]
	Let $u_{k,N}=\phi_{N+k,N}e^{-NQ}$, $0\le k<M$, which are
	orthonormal in $L^2(dA)$.  Choose a fixed collar $U$ of $\partial D$
	on which the exterior maps $\varphi_t$ and the holomorphic boundary
	functions $\mathcal Q_t$ extend for $t$ near $1$.  Locally, with a branch
	of the logarithm, we put
	\[
		\mathcal F_t=t\log\varphi_t+\mathcal Q_t.
	\]
	On the exterior, $\operatorname{Re}\mathcal F_t=\widehat Q_t$ by the obstacle identity
	in Section~\ref{sec:proof}. Inside the collar it denotes the harmonic
	continuation of that exterior solution.  With $U_t$ and $C_t$ as in
	Section~\ref{sec:proof}, the weighted growth law
	\eqref{eq:rule-a-drift} gives, for exterior $z$,
	\[
		\partial_t U_t(z)=\int\log|z-w|\,d\omega_{D_t}(w)
		=\log r_t+\log|\varphi_t(z)|.
	\]
	The contact identity gives $C_t'=-\log r_t$, as in the derivation
	of \eqref{eq:norm-derivative}.  Hence,
	$\partial_t\operatorname{Re}\mathcal F_t=\log|\varphi_t|$.  Holomorphicity and the
	normalization at infinity then give, locally on $U$,
	\[
		\partial_t \mathcal F_t=\log\varphi_t.
	\]

	The uniform planar orthogonal-polynomial expansion
	\cite[Theorems~1.3 and 3.2]{HedenmalmWennman} has leading term
	\[
		u_{k,N}(z)=N^{1/4}\sqrt{\varphi_{t_k}'(z)}\,
		\mathcal B_{0,t_k}(z)\exp\{N(\mathcal F_{t_k}(z)-Q(z))\}
		\bigl(1+O(N^{-1})\bigr),
		\qquad t_k=1+k/N,
	\]
	uniformly in a shrinking boundary collar, for each fixed $k$.
	The leading factor $\mathcal B_{0,t}$ is nonzero and varies smoothly with $t$
	\cite[Theorem~1.5]{HedenmalmWennman}.  The identity
	$\partial_t \mathcal F_t=\log\varphi_t$ and smoothness of the flow give
	$N(\mathcal F_{t_k}-\mathcal F_1)=k\log\varphi_1+O(N^{-1})$ uniformly there.
	The $L^2$ expansion and the normal Gaussian localization give, with
	$\varphi=\varphi_1$ and a smooth cutoff $\chi$ supported in $U$ and
	equal to one near $\partial D$,
	\[
		\|u_{k,N}-\chi\varphi^k u_{0,N}\|_{L^2(dA)}\longrightarrow0.
	\]
	The orbitals have asymptotically all their mass where
	$|\varphi|=1+o(1)$.  Since $M$ is fixed, replacing $\varphi^k$ by
	$e^{ik\Theta}$ gives
	\[
		\|u_{k,N}-v_{k,N}\|_{L^2(dA)}\longrightarrow0,
		\qquad v_{k,N}=\chi u_{0,N}e^{ik\Theta},\quad 0\le k<M.
	\]

	Let $m_N(\theta)d\theta$ be the conformal-angle marginal of
	$|\chi u_{0,N}|^2dA$.  The leading coefficient in
	\cite[Theorem~1.5]{HedenmalmWennman} has squared modulus
	$\pi^{-1/2}\sqrt{\Delta Q}$ on $\partial D$.  At a boundary point
	$p$, with signed normal coordinate $x_\perp$,
	$Q-\operatorname{Re}\mathcal F_1=2\Delta Q(p)x_\perp^2+O(x_\perp^3)$.
	Thus, uniformly in $p$ for bounded $\sqrt N x_\perp$,
	\[
		|u_{0,N}(p+x_\perp\mathbf n_p)|^2dA
		=(1+o(1))\frac{\sqrt N}{\pi^{3/2}}
		|\varphi'(p)|\sqrt{\Delta Q(p)}
		e^{-4N\Delta Q(p)x_\perp^2}\,dx_\perp\,d\ell(p).
	\]
	The $L^2$ expansion controls the integrated tails outside this normal
	scale.  Integration in $x_\perp$ cancels the factor $\sqrt{\Delta Q(p)}$.
	Using $|\varphi'(p)|\,d\ell(p)=d\theta$ we obtain
	\[
		\left\|m_N-\frac1{2\pi}\right\|_{L^1([0,2\pi])}
		\longrightarrow0.
	\]

	The band density relative to $dA^{\otimes M}$ is
	$M!^{-1}|\det[u_{k,N}(z_a)]_{k=0}^{M-1}{}_{a=1}^{M}|^2$.
	By multilinearity and the $L^2$ comparison of adjacent orbitals,
	replacing $u_{k,N}$ by $v_{k,N}$ in the determinant changes the
	normalized wave function by $o(1)$ in $L^2(dA^{\otimes M})$.
	The corresponding densities therefore differ by $o(1)$ in $L^1$.
	After integrating normal coordinates, the latter angular density is
	\[
		\frac1{M!}
		\left|\det[e^{ik\theta_a}]_{k=0}^{M-1}{}_{a=1}^{M}\right|^2
		\prod_{a=1}^M m_N(\theta_a).
	\]
	The angular marginal limit and the Vandermonde determinant formula
	show $L^1$, hence total-variation, convergence to
	\eqref{eq:radial-band-cue}.\qedhere

\end{proof}

\subsection{Finite-size correction for a radial weight}
For a confining radial weight, we can compute the band structure factor
in the polar angle $\arg z$.  Set
\[
	h_{x,N}=\int_{\C}|z|^{2x}e^{-2NQ(z)}\,dA(z).
\]
Compression of $e^{in\arg z}$ to the orbital band gives
\begin{equation}\label{eq:radial-band-structure-factor}
	S_{\mathrm{band}}(n)
	=1-\frac1M\sum_{k=N}^{N+M-n-1}
	\frac{h_{k+n/2,N}^2}{h_{k,N}h_{k+n,N}},
	\qquad 1\le n<M.
\end{equation}
For the Gaussian weight, $h_{x,N}=\Gamma(x+1)/N^{x+1}$, and Stirling
expansion gives, uniformly for $1\le n\le M/2$ when $M^2/N\to0$,
\begin{equation}\label{eq:radial-band-cue-correction}
	\frac{S_{\mathrm{band}}(n)}{n/M}
	=1+\frac{n(M-n)}{4N}+O\!\left(\frac{M^4}{N^2}\right).
\end{equation}
If also $M\to\infty$, averaging over $1\le n\le\lfloor M/2\rfloor$
gives the relative correction $M^2/(24N)+o(M^2/N)$.  For fixed $M\ge2$,
the leading averaged correction is instead
$N^{-1}\lfloor M/2\rfloor^{-1}\sum_{n=1}^{\lfloor M/2\rfloor}n(M-n)/4$.
This coefficient applies to the radial weight and polar angle.  The
correction for the closest-boundary conformal angle of the three-fold
droplet remains to be determined.

\providecommand{\bysame}{\leavevmode\hbox to3em{\hrulefill}\thinspace}
\providecommand{\MR}{\relax\ifhmode\unskip\space\fi MR }
\providecommand{\MRhref}[2]{%
  \href{http://www.ams.org/mathscinet-getitem?mr=#1}{#2}
}
\providecommand{\href}[2]{#2}

\end{document}